\documentclass{article}

\usepackage{fullpage}
\usepackage{amsthm,amssymb,amsmath}  
\usepackage{comment,xspace,enumerate}
\usepackage[capitalise]{cleveref}
\usepackage[utf8]{inputenc}
\usepackage{thmtools}
\usepackage{thm-restate}
\usepackage{authblk}

  \theoremstyle{plain}
  \newtheorem{theorem}{Theorem}
  \newtheorem{lemma}{Lemma}  
    
  \newtheorem{fact}{Fact}
  \newtheorem{observation}{Observation}
  \theoremstyle{definition}
  
  \newtheorem{example}{Example}

\title{Two strings at Hamming distance 1 cannot be both quasiperiodic}
\author[1]{Amihood Amir}
\author[2]{Costas S. Iliopoulos}
\author[2,3]{Jakub Radoszewski\footnote{The author is a Newton International Fellow.}$^{,}$}

\affil[1]{
  Department of Computer Science, Bar-Ilan University, Ramat-Gan 52900, Israel\\
  \texttt{amir@cs.biu.ac.il}
}
\affil[2]{
  Department of Informatics, King's College London, London, UK\\
    \texttt{costas.iliopoulos@kcl.ac.uk}
}
\affil[3]{
Institute of Informatics, University of Warsaw, Warsaw, Poland\\
\texttt{jrad@mimuw.edu.pl}
}

\date{\vspace{-5ex}}

\usepackage{amsthm,amssymb,amsmath}  
  \usepackage{tikz}
  \usepackage[T1]{fontenc}
  \usepackage[utf8]{inputenc}
  \usepackage{comment}
  \usepackage{xspace}
  \usepackage{enumerate}
  \usepackage{listings}
  \usepackage[hidelinks]{hyperref}
  \usetikzlibrary{decorations.pathreplacing,calc,snakes}

\begin{document}
\maketitle

  \begin{abstract}
    We present a generalization of a known fact from combinatorics on words related to periodicity 
    into quasiperiodicity.
    A string is called \emph{periodic} if it has a period which is at most half of its length.
    A string $w$ is called \emph{quasiperiodic} if it has a non-trivial \emph{cover},
    that is, there exists a string $c$ that is shorter than $w$ and such that every position in $w$
    is inside one of the occurrences of $c$ in $w$.
    It is a folklore fact that two strings that differ at exactly one position
    cannot be both periodic.
    Here we prove a more general fact that two strings that differ at exactly one position
    cannot be both quasiperiodic.
    Along the way we obtain new insights into combinatorics of quasiperiodicities.
  \end{abstract}


  \section{Introduction}
  A \emph{string} is a finite sequence of letters over an alphabet $\Sigma$.
  If $w$ is a string, then by $|w|=n$ we denote its length, by $w[i]$ for $i \in \{1,\ldots,n\}$ we denote
  its $i$-th letter, and by $w[i..j]$ we denote a \emph{factor} of $w$ being a string composed
  of the letters $w[i] \ldots w[j]$ (if $i>j$, then it is the empty string).
  A factor $w[i..j]$ is called a \emph{prefix} if $i=1$ and a \emph{suffix} if $j=n$.
  
  An integer $p$ is called a \emph{period} of $w$ if $w[i]=w[i+p]$ for all $i=1,\ldots,n-p$.
  A string $u$ is called a \emph{border} of $w$ if it is both a prefix and a suffix of $w$.
  It is a fundamental fact of string periodicity that a string $w$ has a period $p$ if and only if it has a border of length $n-p$;
  see~\cite{AlgorithmsOnStrings,Lothaire}.
  If $p$ is a period of $w$, $w[1..p]$ is called a \emph{string period} of $w$.
  If $w$ has a period $p$ such that $p \le \frac{n}{2}$, then $w$ is called \emph{periodic}.
  In this case $w$ has a border of length at least $\lceil \frac{n}{2} \rceil$.

  For two strings $w$ and $w'$ of the same length $n$, we write $w =_j w'$ if $w[i]=w'[i]$ for all $i \in \{1,\ldots,n\}\setminus \{j\}$
  and $w[j] \ne w'[j]$.
  This means that $w$ and $w'$ are at Hamming distance 1, where the Hamming distance counts the number of different
  positions of two equal-length strings.
  The following fact states a folklore property of string periodicity that we generalize in this work into string quasiperiodicity.
  For completeness we provide its proof in Section~\ref{sec:main}.

  \begin{fact}\label{fct:periodic}
    Let $w$ and $w'$ be two strings of length $n$ and $j \in \{1,\ldots,n\}$ be an index.
    If $w =_j w'$, then at most one of the strings $w$, $w'$ is periodic.
  \end{fact}

  We say that a string $c$ \emph{covers} a string $w$ ($|w|=n$) if for every position $k \in \{1,\ldots,n\}$
  there exists a factor $w[i..j]=c$ such that $i \le k \le j$.
  Then $c$ is called a \emph{cover} of $w$; see Fig.~\ref{fig:cover}.
  A string $w$ is called \emph{quasiperiodic} if it has a cover of length smaller than $n$.

    \begin{figure}[htpb]
    \begin{center}
    \begin{tikzpicture}[xscale=0.5]
  \foreach \x/\c in {
    1/a,2/a,3/b,4/a,5/a,
    6/b,7/a,8/a,
    9/a,10/a,11/b,12/a,13/a,
    14/a,15/b,16/a,17/a
  }  \draw (\x,0) node[above] {\tt \c};
  \draw[yshift=-0.4cm] (0.7,0.9) -- (0.7,1) -- (5.3,1) -- (5.3,0.9);
  \draw[xshift=3cm,yshift=-0.1cm] (0.7,0.9) -- (0.7,1) -- (5.3,1) -- (5.3,0.9);
  \draw[xshift=8cm,yshift=-0.1cm] (0.7,0.9) -- (0.7,1) -- (5.3,1) -- (5.3,0.9);
  \draw[xshift=12cm,yshift=-0.4cm] (0.7,0.9) -- (0.7,1) -- (5.3,1) -- (5.3,0.9);
\end{tikzpicture}
    \end{center}
    \caption{
      \texttt{aabaa} is a cover of \texttt{aabaabaaaabaaabaa}
    }\label{fig:cover}
    \end{figure}

  A significant amount of work has been devoted to the computation of covers in a string.
  A linear-time algorithm finding the shortest cover of a string was proposed
  by Apostolico et al.\ \cite{DBLP:journals/ipl/ApostolicoFI91}.
  Later a linear-time algorithm computing all the covers of a string was proposed by Moore and Smyth
  \cite{DBLP:conf/soda/MooreS94}.
  Breslauer \cite{DBLP:journals/ipl/Breslauer92} gave an on-line $O(n)$-time algorithm
  computing the cover array of a string of length $n$, that is, an array specifying the lengths
  of shortest covers of all the prefixes of the string.
  Li and Smyth~\cite{DBLP:journals/algorithmica/LiS02} provided a linear-time algorithm
  for computing the array of longest covers of all the prefixes of a string.
  All these papers employ various combinatorial properties of covers.

  Our main contribution is stated as the following theorem.
  As we have already mentioned before, a periodic string has a border long enough to be the string's cover.
  Hence, a periodic string is also quasiperiodic, and Theorem~\ref{thm:main} generalizes Fact~\ref{fct:periodic}.

  \begin{theorem}\label{thm:main}
    Let $w$ and $w'$ be two strings of length $n$ and $j \in \{1,\ldots,n\}$ be an index.
    If $w =_j w'$, then at most one of the strings $w$, $w'$ is quasiperiodic.
  \end{theorem}

  The proof of Theorem~\ref{thm:main} is divided into three sections.
  In Section~\ref{sec:prelim} we restate several simple preliminary observations.
  Then, Section~\ref{sec:aux} contains a proof of a crucial auxiliary lemma which shows
  a combinatorial property of seeds that we use extensively in the main result.
  Finally, Section~\ref{sec:main} contains the main proof.

  \section{Preliminaries}\label{sec:prelim}
  We say that a string $s$ is
  a \emph{seed} of a string $w$ if $|s| \le |w|$ and $w$ is a factor of some string $u$ covered by $s$;
  see Fig.~\ref{fig:seed}.
  Furthermore, $s$ is called a \emph{left seed} of $w$ if $s$ is both a prefix and a seed of $w$.
  Thus a cover of $w$ is always a left seed of $w$, and a left seed of $w$ is a seed of $w$.
  The notion of seed was introduced in \cite{DBLP:journals/algorithmica/IliopoulosMP96}
  and efficient computation of seeds was further considered in \cite{DBLP:conf/cpm/ChristouCIKPRRSW11,DBLP:conf/soda/KociumakaKRRW12}.
    
    \begin{figure}[htpb]
    \begin{center}
    \begin{tikzpicture}[xscale=0.5]
  \foreach \x/\c in {
    2/a,3/b,4/a,5/a,
    6/b,7/a,8/a,
    9/a,10/a,11/b,12/a,13/a,
    14/a,15/b,16/a,17/a,
    18/a
  }  \draw (\x,0) node[above] {\tt \c};
  \draw[yshift=-0.4cm] (1.5,1) -- (5.3,1) -- (5.3,0.9);
  \draw[xshift=3cm,yshift=-0.1cm] (0.7,0.9) -- (0.7,1) -- (5.3,1) -- (5.3,0.9);
  \draw[xshift=8cm,yshift=-0.1cm] (0.7,0.9) -- (0.7,1) -- (5.3,1) -- (5.3,0.9);
  \draw[xshift=12cm,yshift=-0.4cm] (0.7,0.9) -- (0.7,1) -- (5.3,1) -- (5.3,0.9);
  \draw[xshift=16cm,yshift=-0.1cm] (0.7,0.9) -- (0.7,1) -- (2.5,1);
\end{tikzpicture}
    \end{center}
    \caption{
      \texttt{aabaa} is a seed of \texttt{abaabaaaabaaabaaa}
    }\label{fig:seed}
    \end{figure}

  In the proof of our main result we use the following easy observations that are immediate consequences
  of the definitions of cover and seed.

  \begin{observation}\label{obs}
    Consider strings $w$ and $c$.
    \begin{enumerate}[(a)]
      \item\label{obs:cover_period}
        If $c$ is a cover of $w$ and $|c| \ge |w|/2$, then $w$ is periodic with a period $|w|-|c|$.
      \item\label{obs:cover_cover}
        If $c$ is a cover of $w$, then any cover of $c$ is also a cover of $w$.
      \item\label{obs:trivial}
        If $c$ is a seed of $w$, then $c$ is a seed of every factor of $w$ of length at least $|c|$.
      \item\label{obs:extend}
        If $w$ has a period $p$ and a prefix of length at least $p$ that has a cover $c$,
        then $c$ is a left seed of $w$.
    \end{enumerate}
  \end{observation}

  A string $w'$ is called a \emph{cyclic shift} of a string $w$, both of length $n$,
  if there is a position $i \in \{1,\ldots,n\}$ such that $w'=w[i+1..n]w[1..i]$.
  We denote this relation as $w \approx w'$.
  The following obviously holds.

  \begin{observation}\label{obs:cyclic_shift_seed}
    If $w'$ is a cyclic shift of $w$, then $w$ is a seed of $w'$.
  \end{observation}

  \section{Auxiliary Lemma}\label{sec:aux}
  In the following lemma we observe a new property of the notion of seed.
  As we will see in Section~\ref{sec:main}, this lemma encapsulates the hardness of multiple cases in the proof of the main result.

  Before we proceed to the lemma, however, let us introduce an additional notion lying in between periodicity and quasiperiodicity.
  We say that a string $w$ of length $n$ is \emph{almost periodic} with period $p$ if there exists an index
  $j \in \{1,\ldots,n-p\}$ such that:
  $$w[i] = w[i+p] \quad\mbox{for all }i=1,\ldots,n-p,\ i \ne j,\quad\mbox{and }w[j] \ne w[j+p].$$
  In this case we refer to $j$ as the \emph{mismatch position}.
  Furthermore, if $w[1..b] =_j w[n-b+1..n]$ for an integer $b$, we say that each of these factors
  is an \emph{almost border} of $w$ of length $b$ (and again refer to $j$ as the mismatch position).
  We immediately observe the following.

  \begin{observation}\label{obs:almost}
    A string $w$ of length $n$ is almost periodic with period $p$ and mismatch position $j$
    if and only if $w$ has an almost border of length $n-p$ with mismatch position $j$.
  \end{observation}

  \begin{example}
    The following string of length 19:
    \begin{center}
      \texttt{abaab\,ab\underline{a}ab\,ab\underline{b}ab\,abba}
    \end{center}
    is almost periodic with period $p=5$ and mismatch position $j=8$ (the letters at positions $j$ and $j+p$ are underlined).
    Hence, it has an almost border of length 14:
    \begin{center}
      \texttt{abaab\,ab\underline{a}ab\,abba}\ $=_8$\ \texttt{abaab\,ab\underline{b}ab\,abba}.
    \end{center}
  \end{example}

  \begin{lemma}\label{lem:main}
    Let $w$ and $w'$ be two strings of length $n$ and $j \in \{1,\ldots,n\}$ be an index.
    If $w =_j w'$, then $w$ is not a seed of $w'$.
  \end{lemma}
  \begin{proof}
    Assume to the contrary that $w$ is a seed of $w'$.
    Let $u$ be a string covered by $w$ that has $w'$ as a factor.
    Obviously, it suffices to consider two occurrences of $w$ in $u$ to cover all positions of the factor $w'$:
    the leftmost one that covers $w'[n]$
    and
    the rightmost one that covers $w'[1]$.
    Let $\alpha$ be the length of the longest suffix of $w'$ that is a prefix of $w$, and
    let $\beta$ be the length of the longest prefix of $w'$ that is a suffix of $w$
    (these are the so-called longest overlaps between $w'$ and $w$, and between $w$ and $w'$).
    Thus we have $\alpha,\beta>0$ and $\alpha+\beta \ge n$; see Fig.~\ref{fig:pref_suf}.
    From now on we assume that $\alpha \ge n/2$.
    The other case (i.e., $\beta \ge n/2$) is symmetric by reversing the strings $w$ and $w'$.
    Let us denote $p=n-\alpha$.

    \begin{figure}[htpb]
    \begin{center}
    \begin{tikzpicture}
  \draw (0,0) -- (10,0);
  \draw (0.3,0) node[below] {$u$};
  \draw[very thick,xshift=1cm] (2,0) -- (6,0)  (2,-0.1) -- (2,0.1)  (6,-0.1) -- (6,0.1)  (4,0) node[below] {$w'$};
  \draw[yshift=0.7cm,xshift=-1cm,very thick] (2,0) -- (6,0)  (2,-0.1) -- (2,0.1)  (6,-0.1) -- (6,0.1)  (4,0) node[below] {$w$};
  \draw[yshift=0.7cm,xshift=-1cm,latex-latex] (4,0.3) -- node[above] {$\beta$} (6,0.3);
  \draw[yshift=1.8cm,xshift=2cm,very thick] (2,0) -- (6,0)  (2,-0.1) -- (2,0.1)  (6,-0.1) -- (6,0.1)  (4,0) node[below] {$w$};
  \draw[yshift=1.8cm,xshift=2cm,latex-latex] (2,0.3) -- node[above] {$\alpha$} (5,0.3);
\end{tikzpicture}
    \end{center}
    \caption{
      $w$ is a seed of $w'$; $\alpha$ and $\beta$ are the longest overlaps between the two strings.
    }\label{fig:pref_suf}
    \end{figure}
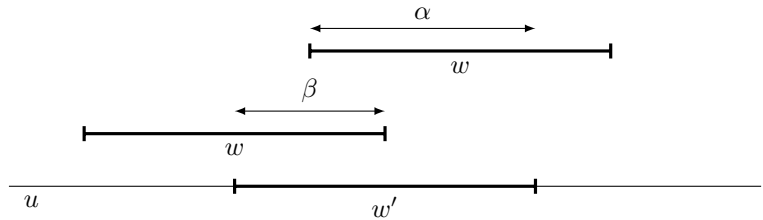

    First consider the case when $j$ satisfies $p < j \le \alpha$.
    Then we have:
    $$w'[1..\alpha] =_j w[1..\alpha] \quad\mbox{and}\quad w[1..\alpha] = w'[p+1..n]$$
    by the definitions of $w'$ and $\alpha$, respectively.
    Consequently, $w'[1..\alpha] =_j w'[p+1..n]$.
    This means that $w'$ has an almost border of length $\alpha$ with mismatch position $j$.
    By Observation~\ref{obs:almost}, $w'$ is almost periodic with period $p$ and the same mismatch position.

    The latter can be written equivalently as follows:
    \begin{equation}\label{eq:w1}
      w'[1..p] \approx w'[2..p+1] \approx \ldots \approx w'[j..j+p-1] \not\approx w'[j+1..j+p] \approx \ldots \approx w'[\alpha+1..n].
    \end{equation}
    Recall that $w =_j w'$.
    This means that the same cyclic-shift relations hold for all corresponding factors of $w$ that do not contain
    the symbol $w[j]$.
    Moreover, $w[1..\alpha] = w'[p+1..n]$, so $w[j] = w'[j+p] = w[j+p]$ (and $w[j] \ne w[j-p]$).
    This concludes that:
    \begin{equation}\label{eq:w}
      w[1..p] \approx w[2..p+1] \approx \ldots \approx w[j-p..j-1] \not\approx w[j-p+1..j] \approx \ldots \approx w[\alpha+1..n].
    \end{equation}
    Moreover, the inequalities satisfied by $j$ in this case imply that $w'[1..p] = w[1..p]$ and $w'[\alpha+1..n] = w[\alpha+1..n]$.
    Hence, the conditions (1) and (2) conclude that there is no suffix of $w$ of length at least $p$ that would be
    a prefix of $w'$.
    Consequently, $\beta<p$, a contradiction.
    
    We are left with two cases:
    \begin{enumerate}[(A)]
      \item $j \le p$.
        In this case $w$ has a border of length $\alpha$:
        $$w[1..\alpha] = w'[p+1..n] = w[p+1..n].$$
        Consequently, $w$ is periodic with period $p$.
        On the other hand, $w'$ does not have the period $p$, since $w'[j+p] = w[j+p] = w[j] \ne w'[j]$.
        Moreover, $w'[1..p] \not \approx w[1..p] \approx w[i..i+p-1]$ for all $i$.
        In conclusion, there cannot exist a suffix of $w$ of length at least $p$ that would be a prefix of $w'$,
        i.e.\ $\beta<p$, a contradiction.

      \item $j>\alpha$.
        In this case $w'$ has a border of length $\alpha$:
        $$w'[1..\alpha] = w[1..\alpha] = w'[p+1..n].$$
        Consequently, $w'$ is periodic with period $p$.
        On the other hand, $w$ does not have the period $p$, since $w[j-p] = w'[j-p] = w'[j] \ne w[j]$.
        Moreover,
        $$w[\alpha+1..n] \not \approx w'[\alpha+1..n] \approx w'[i..i+p-1]$$
        for all $i$.
        In conclusion, there cannot exist a suffix of $w$ of length at least $p$ that would be a prefix of $w'$,
        i.e.\ $\beta<p$, a contradiction.\qedhere
    \end{enumerate}
  \end{proof}

  The following example illustrates the main case of the proof of the above lemma.

  \begin{example}
    Consider the following two strings of length 19:
    \begin{center}
      $w=\,$\texttt{abaab\,ab\underline{b}ab\,abbab\,abba},\quad $w'=\,$\texttt{abaab\,ab\underline{a}ab\,abbab\,abba}.
    \end{center}
    We have $w=_8w'$.
    The longest suffix of $w'$ that is a prefix of $w$ has length $14$ (\texttt{abaab\,abbab\,abba}).
    Hence, $w'$ is almost periodic with period $19-14=5$ and mismatch position 8.
    Moreover, $w$ is almost periodic with the same period and mismatch position $8-5=3$.
    We see that no prefix of $w'$ of length at least 5 can be a suffix of $w$.
  \end{example}

  We use Lemma~\ref{lem:main} as our key tool throughout the proof of the main result.
  As a consequence of Lemma~\ref{lem:main} we obtain the following lemma that will be also useful in the main proof.

  \begin{lemma}\label{lem:same_cover_seed}
    Let $w$ and $w'$ be two strings of length $n$ and $j \in \{1,\ldots,n\}$ be an index.
    If $w =_j w'$, then there does not exist a string $c$ that would be both a cover of $w$
    and a seed of $w'$.
  \end{lemma}
  \begin{proof}
    Consider an occurrence of $c$ in $w$, $w[i..i+|c|-1]$, that covers the position $j$.
    Due to Observation~\ref{obs}\ref{obs:trivial}, the string $c'=w'[i..i+|c|-1]$ has $c$ as a seed.
    We have $c' =_{j-i+1} c$, which contradicts Lemma~\ref{lem:main}.
  \end{proof}

  \section{Main Result}\label{sec:main}
  In this section we first present a proof of the folklore property of string periodicity
  (Fact~\ref{fct:periodic}) for completeness, and then proceed to the proof of our main result being a generalization
  of that fact (Theorem~\ref{thm:main}).

  \begin{proof}[Proof (of Fact~\ref{fct:periodic})]
    Assume to the contrary that $w =_j w'$ and both strings are periodic.
    Let $p$ and $p'$ ($p,p' \le \frac{n}{2}$) be the shortest periods of $w$ and $w'$.
    Assume w.l.o.g.\ that $p \le p'$.
    It suffices to prove the lemma in the case that $w'$ is a square of length $2p'$ and $j \le p'$.
    Let us define $w_1=w[1..p']$ and $w_2=w[p'+1..2p']$.
    By the periodicity of $w'$, we see that $w_1 =_j w_2$.

    We may assume that $j \le \frac{p'}{2}$, as otherwise we may reverse both strings $w_1$, $w_2$.
    Both $w_1$ and $w_2$ have period $p$, as they are factors of $w$ (or the reversal of $w$), and their string periods of length $p$,
    further denoted by $s_1$ and $s_2$, are cyclic shifts.
    Now consider any $k$ such that $kp+1 \le j \le (k+1)p$ (it exists by the upper bound on the value of $j$).
    Then $w_1[kp+1..(k+1)p]=s_1$ and $w_2[kp+1..(k+1)p]=s_2$.
    This concludes that $s_1$ and $s_2$ differ at exactly one position $l=j-kp$, i.e., $s_2=_l s_1$.
    However, $s_1$ is a cyclic shift of $s_2$, hence a seed of $s_2$ by Observation~\ref{obs:cyclic_shift_seed}.
    This contradicts Lemma~\ref{lem:main}.
  \end{proof}

  \begin{proof}[Proof (of Theorem~\ref{thm:main})]
    Assume to the contrary that $w =_j w'$ and both strings are quasiperiodic.
    Let $c$ and $c'$ be the shortest covers of $w$ and $w'$.
    W.l.o.g.\ we can assume that $|c| \le |c'|$.
    We consider a few cases depending on the lengths of the covers:
    \begin{enumerate}[(A)]
      \item $\frac{n}{2} \le |c|$.
        By Observation~\ref{obs}\ref{obs:cover_period}, the strings $w$, $w'$ are both periodic.
        This contradicts Fact~\ref{fct:periodic}.
      \item $|c| < \frac{n}{2} \le |c'|$.
        Again by Observation~\ref{obs}\ref{obs:cover_period}, $w'$ is periodic with the period $p'=n-|c'| \le \frac{n}{2}$.
        Assume w.l.o.g.\ that $j > \frac{n}{2}$.
        This means that the first half of $w$ and $w'$, $w[1..\lfloor n/2\rfloor]$, has period $p'$ and $c$ is its left seed.
        There are three subcases:
        \begin{enumerate}[(B1)]
          \item $p' \le |c|$.
          By Observation~\ref{obs}\ref{obs:extend}, $c$ is a left seed of $w'$.
          Therefore, $w$ and $w'$ contradict Lemma~\ref{lem:same_cover_seed}.
          \item $|c| < p'$ and $p' < j \le p'+|c|$.
          In this case the strings $s=w[p'+1..p'+|c|]$ and $c=w'[p'+1..p'+|c|]=w'[1..|c|]$ differ only at position $j-p'$.
          By Observation~\ref{obs}\ref{obs:trivial} applied to $w$, $s$ has a seed $c$.
          This contradicts Lemma~\ref{lem:main}.
          \item $|c| < p'$ and $j > p'+|c|$.
          Then $c$ is a cover of $w'[1..p'+|c|]$, as it is a left seed of
          $w[1..p'+|c|]=w'[1..p'+|c|]$ and a suffix due to the period $p'$.
          Hence, by Observation~\ref{obs}\ref{obs:extend}, $c$ is a left seed of $w'$.
          Therefore, $w$ and $w'$ contradict Lemma~\ref{lem:same_cover_seed}.
        \end{enumerate}
        From now on we assume that $|c|,|c'| < \frac{n}{2}$.
      \item $c=c'$.
        This immediately contradicts Lemma~\ref{lem:same_cover_seed}.
      \item $|c|=|c'|$ but $c \ne c'$.
        Let $m=|c|$; $m \le n/2$.
        Then $c$ is a border of $w$, and $c'$ is a border of $w'$.
        As $c \ne c'$, it is not possible to change a single position in $w$ such that
        both its prefix and its suffix of length $m$ become $c'$.
      \item $|c|<|c'|$.
        We consider three final subcases.
        \begin{enumerate}[(E1)]
          \item $|c| < j < n-|c|+1$.
            This means that $c$ is a border of $w'$, consequently a border of $c'$.
            However, $c$ is not a cover of $c'$.
            Otherwise, by Observation~\ref{obs}\ref{obs:cover_cover}, $c$ would be a cover of $w'$ shorter than $c'$.

            Consider the factors $w[1..|c'|]$ and $w[n-|c'|+1..n]$;
            note that they cover disjoint sets of positions.

            If $j \le |c'|$, then $w[n-|c'|+1..n]=w'[n-|c'|+1..n]=c'$.
            The string $c$ is a border of $c'=w[n-|c'|+1..n]$ and a cover of $w$.
            Hence, by Observation~\ref{obs}\ref{obs:trivial}, $c$ is a cover of $c'$.
            This contradicts the opposite observation that we have just made.
            Otherwise (if $j > |c'|$) we see that similarly $c$ is
            a cover of $w[1..|c'|] = c'$, again a contradiction.
          \item $j \ge n-|c|+1$.
            This case is symmetric to the following case (D3)
            by reversing the strings $w$ and $w'$.
          \item $j \le |c|$.
            As $c'$ is a prefix of $w'$, this means that the prefix of $c'$
            of length $|c|$ is a string $c_1$ such that $c_1 =_j c$.
            
            Note that $w[n-|c'|+1..n] = w'[n-|c'|+1..n] = c'$.
            The string $c$ is a cover of $w$, therefore, by Observation~\ref{obs}\ref{obs:trivial},
            $c$ is a seed of the prefix $c_1$ of $w[n-|c'|+1..n]$.
            This, however, contradicts Lemma~\ref{lem:main}.
        \end{enumerate}
    \end{enumerate}
    The above cases include all the possibilities.
    This concludes the proof.
  \end{proof}
  
  \section{Conclusions}
  In this note we have proved that every two distinct quasiperiodic strings of the same length differ at more than one position.
  This bound is tight, as, for instance, for every even $n \ge 2$ the strings $\mathtt{a}^{n/2-1}\mathtt{b}\mathtt{a}^{n/2-1}\mathtt{b}$ and $\mathtt{a}^n$
  are both quasiperiodic and differ at exactly two positions.

  \paragraph{Acknowledgements}
  The authors thank Maxime Crochemore and Solon P.\ Pissis for helpful discussions.


  \bibliographystyle{plain}
  \bibliography{cover-1-symbol}

\begin{thebibliography}{1}

\bibitem{DBLP:journals/ipl/ApostolicoFI91}
Alberto Apostolico, Martin Farach, and Costas~S. Iliopoulos.
\newblock Optimal superprimitivity testing for strings.
\newblock {\em Inf. Process. Lett.}, 39(1):17--20, 1991.

\bibitem{DBLP:journals/ipl/Breslauer92}
Dany Breslauer.
\newblock An on-line string superprimitivity test.
\newblock {\em Inf. Process. Lett.}, 44(6):345--347, 1992.

\bibitem{DBLP:conf/cpm/ChristouCIKPRRSW11}
Michalis Christou, Maxime Crochemore, Costas~S. Iliopoulos, Marcin Kubica,
  Solon~P. Pissis, Jakub Radoszewski, Wojciech Rytter, Bartosz Szreder, and
  Tomasz Wale\'n.
\newblock Efficient seeds computation revisited.
\newblock In Raffaele Giancarlo and Giovanni Manzini, editors, {\em
  Combinatorial Pattern Matching - 22nd Annual Symposium, {CPM} 2011.
  Proceedings}, volume 6661 of {\em Lecture Notes in Computer Science}, pages
  350--363. Springer, 2011.

\bibitem{AlgorithmsOnStrings}
Maxime Crochemore, Christophe Hancart, and Thierry Lecroq.
\newblock {\em Algorithms on Strings}.
\newblock Cambridge University Press, New York, NY, USA, 2007.

\bibitem{DBLP:journals/algorithmica/IliopoulosMP96}
Costas~S. Iliopoulos, D.~W.~G. Moore, and Kunsoo Park.
\newblock Covering a string.
\newblock {\em Algorithmica}, 16(3):288--297, 1996.

\bibitem{DBLP:conf/soda/KociumakaKRRW12}
Tomasz Kociumaka, Marcin Kubica, Jakub Radoszewski, Wojciech Rytter, and Tomasz
  Wale\'n.
\newblock A linear time algorithm for seeds computation.
\newblock In Yuval Rabani, editor, {\em Proceedings of the Twenty-Third Annual
  {ACM-SIAM} Symposium on Discrete Algorithms, {SODA} 2012}, pages 1095--1112.
  {SIAM}, 2012.

\bibitem{DBLP:journals/algorithmica/LiS02}
Yin Li and William~F. Smyth.
\newblock Computing the cover array in linear time.
\newblock {\em Algorithmica}, 32(1):95--106, 2002.

\bibitem{Lothaire}
M.~Lothaire.
\newblock {\em Combinatorics on Words}.
\newblock Addison-Wesley, Reading, MA., U.S.A., 1983.

\bibitem{DBLP:conf/soda/MooreS94}
Dennis Moore and William~F. Smyth.
\newblock Computing the covers of a string in linear time.
\newblock In {\em SODA}, pages 511--515, 1994.

\end{thebibliography}

\end{document}